\documentclass[letterpaper, 10 pt, conference]{ieeeconf}  

\IEEEoverridecommandlockouts                              

\usepackage{amssymb}
\usepackage{graphicx}
\usepackage{epstopdf}
\usepackage{amsmath}
\usepackage{mathtools, cuted}
\usepackage{multirow}
\usepackage{subfigure}
\usepackage{array}
\usepackage{bm}
\usepackage{url}

\usepackage{xcolor}

\usepackage{algorithm}
\usepackage{algpseudocode}

\usepackage{tabularx}
\usepackage{stackengine}
\setstackEOL{\\}

\newtheorem{problem}{\textbf{Problem}}
\newtheorem{definition}{\textbf{Definition}}

\newtheorem{lemma}{\textbf{Lemma}}
\newtheorem{corollary}{\textbf{Corollary}}
\newtheorem{remark}{\textbf{Remark}}
\newtheorem{proposition}{\textbf{Proposition}}

\usepackage[noadjust]{cite}

\title{\LARGE \bf Detection and Inference of Randomness-based Behavior for Resilient Multi-vehicle Coordinated Operations
}
\author{Paul J Bonczek and Nicola Bezzo
\thanks{Paul J Bonczek and Nicola Bezzo are with the Charles L. Brown Department of Electrical and Computer Engineering, and Link Lab, University of Virginia, Charlottesville, VA 22904, USA. Email: {\tt \{pjb4xn, nb6be\}@virginia.edu}}
}

\newcommand*{\N}{\mathbb{N}}
\newcommand*{\R}{\mathbb{R}}
\newcommand*{\E}{\mathbb{E}}

\renewcommand{\baselinestretch}{0.964}
\begin{document}

\bstctlcite{IEEEexample:BSTcontrol}

\maketitle
\thispagestyle{empty}
\pagestyle{empty}
\begin{abstract}

A resilient multi-vehicle system cooperatively performs tasks by exchanging information, detecting, and removing cyber attacks that have the intent of hijacking or diminishing performance of the entire system. In this paper, we propose a framework to: i) detect and isolate misbehaving vehicles in the network, and ii) securely encrypt information among the network to alert and attract nearby vehicles toward points of interest in the environment without explicitly broadcasting safety-critical information. To accomplish these goals, we leverage a decentralized virtual spring-damper mesh physics model for formation control on each vehicle. To discover inconsistent behavior of any vehicle in the network, we consider an approach that monitors for changes in sign behavior of an inter-vehicle residual that does not match with an expectation. Similarly, to disguise important information and trigger vehicles to switch to different behaviors, we leverage side-channel information on the state of the vehicles and characterize a hidden spring-damper signature model detectable by neighbor vehicles. Our framework is demonstrated in simulation and experiments on formations of unmanned ground vehicles (UGVs) in the presence of malicious man-in-the-middle communication attacks.

\end{abstract}
\section{Introduction} \label{sec:introduction}

The use of coordinated multi-vehicle systems to perform various tasks has been extensively explored for many years \cite{search_and_rescue,surveillance1,military_convoy,navigation}. By leveraging multiple vehicles instead of only one, it is possible to perform more operations, and complete a task faster and more efficiently. Examples of such operations that can benefit from the use of multi-vehicle systems are search and rescue operations \cite{search_and_rescue} depicted in Fig.~\ref{fig:intro}, surveillance \cite{surveillance1}, military convoying/platooning \cite{military_convoy}, and exploration missions \cite{navigation}.
Generally, approaches that leverage multi-vehicle systems assume that all vehicles are cooperative while performing the desired operations to maintain swarming formations and can exchange all necessary information to achieve the desired goal. However, these vehicles are susceptible to malicious external attacks, especially on their communication infrastructure, which can affect the entire network performance. For example, with a Man-In-The-Middle (MITM) attack \cite{MITM}, an attacker intercepts a communication broadcast and replaces it with altered data which are then received by neighboring vehicles. Successful attackers are able to purposefully block important information from being received by nearby vehicles in the formation or control the entire multi-vehicle network to an undesired location.

\begin{figure}[t]
\centering
\includegraphics[width=0.48\textwidth]{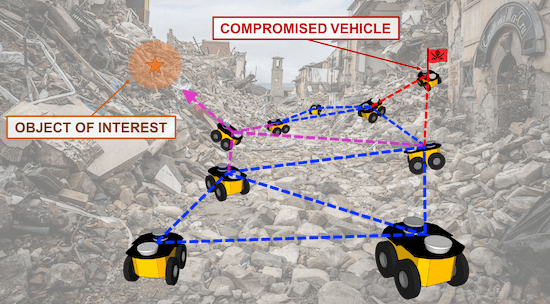}
\vspace{-10pt}
\caption{Pictorial motivation of the problem in this paper in which a multi-vehicle system cooperatively performs a task while inferring the objective of other teammates and detecting if they are compromised by cyber-attacks.}
\vspace{-16pt}
\label{fig:intro}
\end{figure}

Safety-critical information, if not properly encrypted, can also be intercepted creating further security issues. Although encryption techniques can be deployed, there exists attacks that are capable of discovering encryption keys to extract data. The most secure option is to avoid exchanging data altogether. In this work we propose to leverage side-channel information that contains hidden data, which is unknown to malicious attackers. For example, if a vehicle discovers an object of interest whose identity needs to be kept secret, as depicted in Fig.~\ref{fig:intro}, it could perform a certain signature motion (similar to watermarking \cite{watermarking}) to indicate to neighboring vehicles of the discovered object. This motion triggers the surrounding vehicles to infer its position and switch tasks to get attracted to the same object. In this way, a vehicle can collect data and infer the behavior of other vehicles without explicitly broadcasting important information. With such premises, in this work we focus on applications for cooperative autonomous vehicle networks in the presence of adversaries. We expand on literature that leverage virtual springs for decentralized formation control \cite{Energy_Conserve,CHEN20161730,BEZZO201494,6319411} while introducing a monitoring approach to detect inconsistent behaviors between expected and received data to provide: 1) resiliency to cyber-attacks on communication broadcasts and 2) discovery of a hidden signature via side-channel states.

\subsection{Related Work} \label{sec:related_work}

Analyzing the literature on the area of multi-vehicle resilience, we find works employing the Mean Subsequence Reduced (MSR) algorithm that provides resiliency to $F$ number of uncooperative agents while still reaching a consensus of a specific value
\cite{7822915,10.1093/imamci/dnx062}. While this is the standard resilient consensus algorithm to agree on a specific value (e.g., heading angle) for teams of mobile robots, it does not provide resilience to misbehaving vehicles within time-varying proximity-based formations.

Similar to previous literature, our proposed work leverages a residual-based detection technique for attack detection \cite{BadData,SPRT,GLR,CUSUM1,Paul_IFAC,Paul_ACC}. Authors in \cite{MVS_residual} leverage the residual-based \textit{Cumulative Sum} (CUSUM) detection procedure to discover spoofs to on-board navigation systems of robots in multi-vehicle systems. In our previous work, we have characterized the \textit{Cumulative Sign} (CUSIGN) detector \cite{Paul_IFAC} on a single vehicle, which is designed to detect non-random (i.e., inconsistent) \textit{signed} residual behavior.
We demonstrated the effectiveness of the randomness-based CUSIGN detector when compared to the magnitude-based CUSUM in the presence of \textit{stealthy} sensor attacks that intentionally hide within noise profiles to remain undetected. Moreover, the fundamental component we want to convey is that noisy systems will follow an expected model behavior under nominal conditions, whereas systems that experience hijacking attempts from an attacker will exhibit contradictory behavior.

In this work, we extend our recent randomness-based detection techniques for sensor spoofing introduced in \cite{Paul_IFAC,Paul_ACC} to detect attacks and hidden signatures in multi-vehicle systems. Specifically, each vehicle monitors the \textit{inter-vehicle residual} ---defined as the difference between received information and predicted values--- using the CUSIGN detector to determine whether nearby vehicles are behaving as expected or not. Additionally, we propose a detection scheme to monitor a residual sign switching rate (i.e., the frequency of residual sign changes) to identify if nearby vehicles are displaying hidden signature behavior, by leveraging known stochastic properties of the system models.

To summarize, the objective of this work is: 1) to detect stealthy MITM attacks on communication broadcasts that leave behind inconsistent inter-vehicle residual behavior, and 2) to provide a method for the network to resiliently maintain operations, while 3) using hidden side-channels to communicate the discovery of an object to nearby vehicles without explicitly sending this information. The contribution of this work is twofold: 1) a detector for discovering inconsistent behavior from stealthy MITM communication attacks in multi-vehicle systems, and 2) a `side-channel'-based scheme to produce and detect a hidden signature to protect critical information from being intercepted in communication broadcasts by attackers in multi-vehicle operations. 
\section{Preliminaries} \label{sec:preliminaries}

Let us consider a multi-vehicle network of $N$ homogeneous robots modeled as a directed graph $\mathcal{G} = (\mathcal{V}, \mathcal{E})$, where we denote $\mathcal{V} = \{1,\dots,N \}$ as the vehicle set and the edge set $\mathcal{E} \subset \mathcal{V} \times \mathcal{V}$, such that an edge $(i,j) \in \mathcal E$ indicates a connection from vehicle $i \in \mathcal V$ to vehicle $j \in \mathcal V$. 
All vehicles are considered to have second order dynamics that can be represented in a linear time-invariant (LTI) state space form:
\begin{equation} \label{eq:vehicle_dynamical_model}
    \dot{\bm{x}}_{i} = \bm{A} \bm{x}_i + \bm{B} \bm{u}_{i} + \bm{\nu}_i, \;\;\; \forall i \in \mathcal{V},
\end{equation}
where $\bm{A}$ and $\bm{B}$ denote state and input matrices, the state vector $\bm{x}_i \in \R^n$ consisting of positions $\bm{p}_i$ and velocities $\bm{v}_i = \dot{\bm{p}}_i$, and $\bm{\nu}_i \in \R^{n}$ representing zero-mean Gaussian process noise. Each vehicle $i \in \mathcal{V}$ within the vehicle network is controlled by a virtual spring-damper physics model as,
\begin{align} \label{eq:spring_force}
\bm{u}_i = \ddot{\bm{p}}_i = & \bigg[ \sum_{j \in \mathcal{S}_i} \kappa_{v}( l_{ij} - l^0_v ) \vec{\bm{d}}_{ij} - \sum_{o \in \mathcal{O}_i} \kappa_{o}( l_{io} - l^0_o ) \vec{\bm{d}}_{io}, \nonumber \\[-1pt]
&+ \;\kappa_{g} l_{ig} \vec{\bm{d}}_{ig} \bigg] - \gamma_v \dot{\bm{p}}_i \in \R^{m},
\end{align}
where $\mathcal{S}_i \subset \mathcal{V}$ is the neighbor set of a vehicle $i$, $ \mathcal{O}_i$ denotes the set of nearby obstacles, while $l^0_v$ and $l^0_o$ are desired rest lengths between the vehicle $i$ and its neighbors and obstacles. The variables $l_{ij}$, $l_{io}$, $l_{ig}$ represent euclidean distances (i.e., virtual spring lengths) and $\kappa_v$, $\kappa_o$, $\kappa_g$ are spring constants between neighboring vehicles, obstacles, and the goal, respectively, while $\vec{\bm{d}}$ denotes the unit vector indicating direction of the forces. Given damping coefficients that satisfy $\gamma_v > 0 $, the multi-vehicle system emulates a true spring-mass mesh where dissipating forces act against the velocities, leading to an equilibrium state of zero velocity in the absence of external forces. All vehicles are fitted with a range sensor providing $360$ degree field of view with a limited range $\delta_r > 0$ for obstacle avoidance. Any vehicle $i$ that comes within sensing range of an obstacle $o \in \mathcal{O}_i$ (with position $\bm{p}_o $) attaches a spring to it.

\subsection{Connected Proximity-based Graph} \label{sec:Graph_model}

In order for the vehicle network to cooperatively maintain the desired proximity-based formation in \eqref{eq:spring_force}, the vehicles broadcast information that is received by any other vehicle within a maximum communication range $\delta_c > 0$.

\begin{definition}[Communication Graph] \label{def:comm_graph}
    Given the $N$ vehicles in set $\mathcal V$ with a maximum communication range $\delta_c$, we define the graph $\mathcal{G}_{\mathcal{C}} = ( \mathcal{V}, \mathcal{E}_{\mathcal{C}} )$ with the following edge set,
    \begin{equation} \label{eq:comm_edges}
        \mathcal{E}_{\mathcal{C}} = \big\{ (i,j) \; \big| \; \big\| \bm{p}_i - \bm{p}_j \big\| \leq \delta_c, \; i,j \in \mathcal{V} \big\},
    \end{equation}
    as the \textit{communication graph} of the vehicle set $\mathcal{V}$.
\end{definition}

Consequently, the set of all vehicles within communication range of a vehicle $i$, denoted as $\mathcal{C}_i \subseteq \mathcal{V}$, follows,
\begin{equation} \label{eq:comm_set}
    \mathcal{C}_i = \big\{ j \in \mathcal{V} \; \big| \; (i,j) \in \mathcal{E}_{\mathcal{C}} \big\}.
\end{equation}

All $N$ vehicles are assumed to be equipped with localization/pose sensors represented in the output vector $\bm{y}^{(k)}_i$ by,
\begin{equation} \label{eq:output_vector}
    \bm{y}^{(k)}_i = \bm{C} \bm{x}^{(k)}_i + \bm{\eta}^{(k)}_i \in \R^{N_s}, \;\;\; \forall i \in \mathcal{V},
\end{equation}
where $\bm{C}$ is the output matrix and $\bm{\eta}^{(k)}_i \in \R^{N_s}$ denotes zero-mean Gaussian measurement noise at every discrete time iteration $k \in \N$. A standard Kalman Filter with gain $\bm{K} \in \R^{n \times N_s}$ provides a state estimate $\hat{\bm{x}}_i^{(k)} \in \R^{n}$. To enable proximity-based formation control, each vehicle $i\in \mathcal V$ broadcasts its position estimate $\hat{\bm{p}}_i^{(k)}$ (within the state estimate vector) that is received by any nearby vehicles $j \in \mathcal{C}_i$.
The neighbor set $\mathcal S_i$ in \eqref{eq:spring_force} is used to control the motion of each vehicle $i$ and is computed following Gabriel Graph rule \cite{GabrielGraph1},
\begin{equation} \label{eq:neighbor_set}
    \mathcal{S}_i = \big\{ j \in \mathcal{V} \setminus \mathcal{R}_i \; \big| \; \widehat{ihj} \leq \pi/2, \; j,h \in \mathcal{C}_i \big\},
\end{equation}
where $\widehat{ihj}$, $i\neq j \neq h$ is the interior angle within a three vehicle configuration obtained from the on-board position estimate $\hat{\bm{p}}_i^{(k)}$ and received position estimates $\hat{\bm{p}}_j^{(k)}$ and $\hat{\bm{p}}_h^{(k)}$ from vehicles $j,h \in \mathcal{C}_i$. The set $\mathcal{R}_i \subset \mathcal{V}$ is a subset of vehicles that are deemed compromised by vehicle $i$ and not included in the control graph. 

\begin{definition}[Control Graph] \label{def:proximity_graph}
    Given the vehicle set $\mathcal V$ with each vehicle $i \in \mathcal{V}$ having a neighbor set for control $\mathcal{S}_i \subseteq \mathcal{C}_i$ computed from the Gabriel Graph rule in \eqref{eq:neighbor_set}, we define the graph $\mathcal{G}_{\mathcal{U}} = ( \mathcal{V}, \mathcal{E}_{\mathcal{U}} )$ with the edge set,
    \begin{equation} \label{eq:control_edges}
        \mathcal{E}_{\mathcal{U}} = \big\{ (i,j) \; \big| \; j \in \mathcal{S}_i, \forall i \in \mathcal{V} \big\},
    \end{equation}
    as the \textit{control graph} of the vehicle set $\mathcal{V}$.
\end{definition}

Construction of the control graph by leveraging the Gabriel Graph rule \cite{GabrielGraph1} allows for a connected graph without crossing edges and a uniform coverage (while maintaining desired distances between vehicles) of the network \cite{Energy_Conserve,CHEN20161730,BEZZO201494,6319411}.

\subsection{Attack Model} \label{sec:Attack_model}

We assume the multi-vehicle network is navigating within an adversarial environment, such that individual vehicles may be subject to malicious communication attacks (e.g., MITM attacks \cite{MITM}). In the case of an attack on an unprotected proximity-based formation, a single compromised vehicle can affect the entire network of $N$ vehicles as the effects of the attack are propagated throughout the network. During a persistent communication attack, we assume that an attacker can continuously intercept and modify broadcast data with stealthy (i.e., hidden within the system noise profile)
information in an attempt to intentionally fool (i.e., hijack) the vehicle network. Each vehicle $i$ exchanges state estimates, nearby obstacle positions, and neighbor set information at every time instance $k$ such that nearby vehicles have knowledge of its intended motion by construction of the network model in \eqref{eq:spring_force}. We indicate the spoofed broadcast information from a vehicle $i \in \mathcal{V}$ that is received by other vehicles as:
\begin{equation} \label{eq:spoofed_information}
\begin{split}
    \hat{\bm{x}}_i^{(k)} + \bm{\xi}_i^x & \longrightarrow \Tilde{\hat{\bm{x}}}_i^{(k)}, \\
    \bm{p}_{o} + \bm{\xi}_i^o & \longrightarrow \Tilde{\bm{p}}_{o}, \; \forall o \in \mathcal{O}_i, \\
    \{ \mathcal{S}_i \setminus \mathcal{S}_{i}^{\xi^-} \} \cup \mathcal{S}_{i}^{\xi^+} & \longrightarrow \Tilde{\mathcal{S}}_{i},
\end{split}
\end{equation}
where $\bm{\xi}_i^x \in \R^n$ and $\bm{\xi}_i^o \in \R^{2}$ denote the attack vectors on state and obstacle positions, whereas the sets $\mathcal{S}_{i}^{\xi^-} \hspace{-3pt} \subset \mathcal{V}$ and $\mathcal{S}_{i}^{\xi^+} \hspace{-3pt} \subset \mathcal{V}$, $ \hspace{2pt} \big\{ \mathcal{S}_{i}^{\xi^-} \hspace{-2pt} \cap \hspace{1.2pt} \mathcal{S}_{i}^{\xi^+} \hspace{-1pt} \big\} \hspace{-1pt} = \emptyset$ are vehicle identifications that are removed from and added to the original neighbor set $\mathcal{S}_{i}$, respectively. For any attack vector $\bm{\xi}_i^x \ne 0$, $\bm{\xi}_i^o \ne 0, \; \forall o$, or sets satisfying $| \mathcal{S}_{i}^{\xi^+}| , | \mathcal{S}_{i}^{\xi^-} | > 0$, an attacker is replacing the original message such that the received information by any nearby neighbors will differ from the intended broadcast.

\subsection{Problem Formulation} \label{sec:problem}

Given the network described by the virtual spring model \eqref{eq:spring_force} and the \textit{control graph} $\mathcal{G}_{\mathcal{U}}(\mathcal{V},\mathcal{E}_{\mathcal{U}})$, we are interested in solving the following problems:
\begin{problem}[Vehicle Inconsistency Detection] \label{problem1}
Create a decentralized detection policy $\mathcal{P}_d$ such that a vehicle $j \in \mathcal{V}$ that is experiencing inconsistent behavior can be discovered and isolated by any vehicle $i \in \mathcal{V}$ such that,
\vspace{-2pt}
\begin{equation} \label{prob:isolate_remove_edge}
    (i,j) \notin \mathcal{E}_{\mathcal{U}}, \; i \ne j,
\end{equation}
to prevent undesirable effects to the multi-vehicle network.
\end{problem}

A second problem that we explore in this work is to enable indirect exchange of information by leveraging signature mobility behaviors of the agents of the swarm.  While navigating through an adversarial environment, vehicles that come into sensing range of an object of interest desire to notify the remaining vehicles in the network of their discovery without revealing explicitly the identification and position of the object to maintain secrecy from adversaries.

\begin{problem}[Hidden Signature Detection]
\label{problem2}
Given a vehicle $i \in \mathcal{V}$ that has found an object of interest while navigating within an environment, find a control policy $\mathcal{P}_u$ to covertly provide an identifiable hidden signature $\bm{u}_i^{\mathcal{H}} \in \R^m$ for any nearby vehicles $j \in \mathcal{C}_i \subset \mathcal V$ to detect without explicitly sending information of the discovered object through communication broadcasts.
\end{problem}

Upon recognizing a signature behavior, neighbors of the vehicle will estimate the position of the object based on the same signature and switch toward that object. 

\section{Framework} \label{sec:framework}

In this section we describe the decentralized monitoring framework for detection and isolation of inconsistently behaving vehicles in the network, while allowing each vehicle to provide a hidden signature for nearby vehicles. The diagram in Fig.~\ref{fig:architecture} summarizes our proposed scheme in which each vehicle follows the primary or hidden control model, as well as detects whether neighboring robots have expected behavior according to the primary or hidden models.
\vspace{-3pt}
\begin{figure}[th!b]
\hspace{-4pt}
\centering
\includegraphics[width=0.48\textwidth]{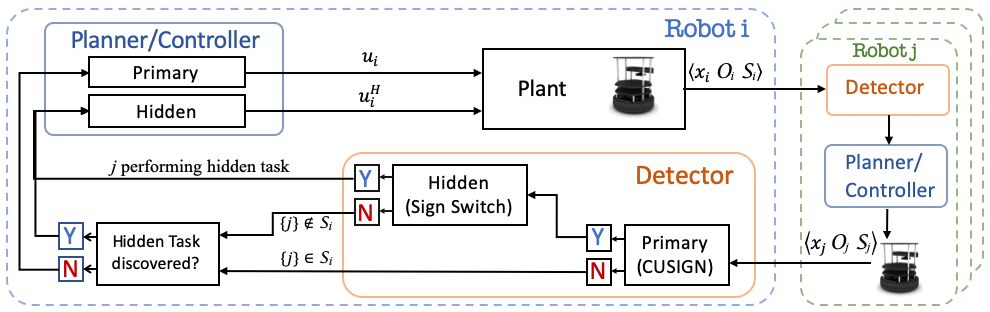}
\vspace{-7pt}
\caption{Overall framework architecture followed by each vehicle $i \in \mathcal{V}$.}
\vspace{-6pt}
\label{fig:architecture}
\end{figure}

\subsection{Monitoring Vehicles for Consistency} \label{sec:Neighbor_monitoring}

During operations, each vehicle monitors nearby vehicles for consistent behavior according to the network model described in \eqref{eq:spring_force}. Each vehicle $i$ receives broadcast information from any nearby vehicle $j \in \mathcal{C}_i$ as represented in \eqref{eq:comm_set}. This vehicle $i$ is able to make state evolution predictions of a nearby vehicle $j \in \mathcal{S}_i$ such that the neighbor set of vehicle $j$ satisfies $\mathcal{S}_j \subset \mathcal{C}_i$. The inclusion of the neighbor set $\mathcal{S}_j \subset \mathcal{C}_i$ is needed in order for vehicle $i$ to predict the future state of the system using \eqref{eq:spring_force}. The state prediction of a vehicle $j$ computed by a vehicle $i$ is computed as,
\begin{equation} \label{eq:neighbor_predict}
    \bar{\bm{x}}_{ij}^{(k+1)} = \bm{A} \hat{\bm{x}}_{j}^{(k)} + \bm{B} \bm{u}_{ij}^{(k)} \in \R^n,
\end{equation}
where $\bm{u}_{ij}^{(k)} \in \R^m$ is the estimated input for vehicle $j$ that is computed by vehicle $i$ which follows the primary network model \eqref{eq:spring_force}. At every $k$th time iteration, a vehicle $i$ compares the \textit{inter-vehicle residual} $\bm{r}_{ij}^{(k)}$ ---defined as the difference between the received state information $\hat{\bm{x}}_j^{(k)}$ and the computed state prediction of a vehicle $j \in \mathcal{S}_i$--- by,
\begin{equation} \label{eq:prediction_residual}
    \bm{r}_{ij}^{(k)} = \hat{\bm{x}}_j^{(k)} - \bar{\bm{x}}_{ij}^{(k)} \in \R^n.
\end{equation}

If a vehicle $j$ is attack-free and is following the primary network model while monitored by a vehicle $i$, each element $q \in \{ 1, \dots, n \}$ of the inter-vehicle residual vector is normally distributed $r_{ij,q}^{(k)} \sim \mathcal{N} \big(0,\sigma_{r,q}^2 \big)$ described as follows,
\begin{equation} \label{eq:update_residual_distribution}
    \E[r_{ij,q}] = 0, \;\;\;\; \mathrm{Var}[r_{ij,q}] = \sum_{s=1}^{N_s} \Big( K_{(q,s)} \sigma_{z,s} \Big)^2 ,
\end{equation}
where $ K_{(q,s)}$ represents the element of the $q$th row and $s$th column of the steady state Kalman gain $\bm{K}$ discussed in Section \ref{sec:Graph_model}. The variable $\sigma_{z,s}^2$ is the $s$th diagonal element of the measurement residual covariance matrix $\bm{\Sigma}_z \in \R^{N_s \times N_s}$ from the on-board state estimation process with $N_s$ sensors (see works \cite{CUSUM1,Paul_ACC,Paul_IFAC} for details of the measurement residual characteristics). Since the network consists of $N$ homogeneous vehicles, then all vehicles share the same values $\bm{K}$ and $\bm{\Sigma}_z$. Each $q$th element of $\bm{r}_{ij}^{(k)}$ is a zero-mean normally distributed variable that is characterized as:
\begin{equation}
\label{eq:binomial_probs}
    \begin{split}
    \Pr\big(r_{ij,q}^{(k)} < \E[r_{ij,q}^{(k)}] \big) &= p_- = 0.5, \\
    \Pr\big(r_{ij,q}^{(k)} > \E[r_{ij,q}^{(k)}] \big) &= p_+ = 0.5,
    \end{split}
\end{equation}
during nominal (i.e., no attack) conditions.

To monitor whether the incoming information from nearby vehicles is behaving in an expected random manner with respect to the primary network model \eqref{eq:spring_force}, we employ the Cumulative Sign (CUSIGN) detector \cite{Paul_IFAC} to check for randomness with the following procedure:

\vspace{6pt}
\centerline{\textbf{CUSIGN Detector}}
\vspace{2pt}
\hrule
\vspace{-3pt}
  \begin{equation}
  \label{pro:CUSIGN}
  \small
      \begin{array}{ll}
        \hspace{-11pt} \textbf{Initialize: } S^{(0),+}_{ij,q} \hspace{-2pt} = S^{(0),-}_{ij,q} \hspace{-2pt} = 0, \; \forall i,j,q  & \\[2pt]
        \hspace{-11pt} S^{(k),+}_{ij,q} \hspace{-2pt} = \hspace{-1pt} \max \hspace{-2pt} \big( 0,S^{(k-1),+}_{ij,q} \hspace{-2pt} + \hspace{-1pt} \text{sgn}(r_{ij,q}^{(k)}) \big) \hspace{-1pt} , &  \\[2pt]
        \hspace{-11pt} S^{(k),+}_{ij,q} \hspace{-2pt} = \hspace{-1pt} 0 \text{ and Alarm } \zeta^{(k),+}_{ij,q} = 1, & \hspace{-4pt} \textbf{if } S^{(k-1),+}_{ij,q} \hspace{-2pt} = \hspace{-1pt} \tau, \\[8pt]
        \hspace{-11pt} S^{(k),-}_{ij,q} \hspace{-2pt} = \hspace{-1pt} \min \hspace{-2pt}  \big( 0,S^{(k-1),-}_{ij,q} \hspace{-2pt} + \hspace{-1pt} \text{sgn}(r_{ij,q}^{(k)}) \big) \hspace{-1pt} , &  \\[2pt]
        \hspace{-11pt} S^{(k),-}_{ij,q} \hspace{-2pt} = \hspace{-1pt} 0 \text{ and Alarm } \zeta^{(k),-}_{ij,q} = 1, & \hspace{-4pt} \textbf{if } S^{(k-1),-}_{ij,q} \hspace{-2pt} = \hspace{-1.5pt} -\tau.
      \end{array}
      \normalsize
  \end{equation}
  \vspace{-2.5pt}
\hrule

\vspace{8pt}

The multi-vehicle detection procedure on a vehicle $i$ accumulates the signed values of the inter-vehicle residual in the CUSIGN test variables for a vehicle $j$ and triggers an alarm $\zeta^{(k),\pm}_{ij,q} = 1$ when a user-defined threshold $\tau \in \N$ is reached, otherwise $\zeta^{(k),\pm}_{ij,q} = 0$. As either of the test variables reach their respective thresholds, the test variable is then reset back to zero. 
The alarms for each $q$th element are then sent to a Memoryless Runtime Estimator (MRE) \cite{Paul_IFAC} to provide a run-time update for alarm rates $\hat{A}_{ij,q}^{(k),-}$ and $\hat{A}_{ij,q}^{(k),+}$, for simplicity denoted as $\hat{A}_{ij,q}^{(k),\pm}$, at a time $k$ by the following,
\begin{equation} \label{eq:MRE_algorithm}
    \hat{A}_{ij,q}^{(k),\pm} = \hat{A}_{ij,q}^{(k-1),\pm} + \frac{\big[\zeta_{ij,q}^{(k),\pm} - \hat{A}_{ij,q}^{(k-1),\pm} \big]}{\ell},
\end{equation}
where $\zeta_{ij,q}^{(k)}$ is the alarm, $\ell \geq 10$ is a ``pseudo-window" length, and $\hat{A}_{ij,q}^{(0)}= \E[A^{\pm}]$ is the expected alarm rate. The following lemma provides an expected alarm rate for a vehicle that is free from attacks (i.e., behaving nominally).

\begin{lemma} \label{lem:CUSIGN_exp_AR} 
Given a vehicle $i \in \mathcal{V}$ with a CUSIGN detector \eqref{pro:CUSIGN} with a threshold $\tau \in \N $ that is monitoring a vehicle $j \in \mathcal{V}$ during attack-free conditions, then the inverse of the first element of the following vector,
\begin{equation}
\label{eq:mu}
    \bm{\mu}^{+} = (\bm{I}_{\tau} - \mathcal{Q^{+}})^{-1}\bm{1}_{\tau \times 1} = (\mu_1^{+}, \dots, \mu_{\tau}^{+})^{\mathsf{T}},
    \vspace{-1pt}
\end{equation}
is the expected alarm rate $\E[A^{+}]$, and $\mathcal{Q^{+}} \in \R^{\tau \times \tau} $ represents the transient states of a designed Markov transition matrix.
\end{lemma}

\begin{proof}
    See \cite{Paul_IFAC} for a similar proof.
\end{proof}

\begin{remark}
    The previous lemma describes the expected rate $\E[A^{+}]$ at which the CUSIGN test variable $S^{(k),+}_{ij,q}$ reaches the defined threshold value $\tau$ to trigger an alarm $\zeta_{ij,q}^{(k),+} = 1$. Similarly, the design of a transition matrix with fundamental matrix $\mathcal{Q}^-$ and expected alarm rate $\E[A^-] = (\mu^-_1)^{-1}$ for the negative case is computed with transition probability ($p_+$ and $p_-$) signs inverted. For construction of $\mathcal{Q}^+$ and $\mathcal{Q}^-$, see \cite{Paul_IFAC}.
\end{remark}
\begin{proposition} \label{pro:CUSIGN_distribution}
    Assuming a vehicle $j \in \mathcal{V}$ is not experiencing MITM attacks while being monitored by a vehicle $i \in \mathcal{V}$ and using the MRE algorithm \eqref{eq:MRE_algorithm} for alarm rate estimation, the alarm rate is normally distributed by $\hat{A}_{ij,q}^{\pm}~\sim~\mathcal{N}\Big( \E[A^{\pm}],\frac{\theta \E[A^{\pm}](1-\E[A^{\pm}])}{2\ell-1}\Big)$, where $\theta \in \R_{+}$ is a scaling constant (see \cite{Paul_IFAC}).
\end{proposition}

By leveraging $\E[A^{\pm}] $ in Lemma \ref{lem:CUSIGN_exp_AR}, the following corollary provides detection bounds for the CUSIGN alarm rate. 

\begin{corollary} \label{cor:CUSIGN_bounds}
     Given the $q$th element of the inter-vehicle residual \eqref{eq:prediction_residual} being monitored by CUSIGN \eqref{pro:CUSIGN}, detection of attacks occur for a chosen level of significance $\alpha \in (0,1)$ when the alarm rate no longer satisfies detection bounds (i.e., $\hat{A}_{ij,q}^{(k),\pm} \not\in [\Omega_{-} , \Omega_{+}]$) where $\Omega_{\pm} = \E[A^{\pm}] \pm \Phi^{-1}\big( \alpha/2 \big) \sqrt{ \mathrm{Var}[A^{\pm}] }$, such that $\Phi^{-1}( \cdot )$ is the inverse cumulative distribution function of a normal distribution.
\end{corollary}

\begin{proof}
    See \cite{Paul_IFAC} for a similar proof.
\end{proof}

A vehicle $i$ that detects non-random (i.e., inconsistent) inter-vehicle residual behavior from a vehicle $j$, responds by placing vehicle $j$ in its compromised set $j \in \mathcal{R}_i$, $\mathcal{R}_i \subset \mathcal{V}$, hence removing it from the control graph (i.e., $(i,j) \not\in \mathcal{E}_{\mathcal{U}}$).

\subsection{Hidden Signature Detection} \label{sec:Hidden_Model}

During operations, vehicles are tasked to converge to observed objects of interest while navigating through the environment. As an $i$th vehicle comes within sensing distance $\delta_r$ of the on-board range sensor with respect to an object,
\begin{equation}
    l_{ip} = \| \bm{p}_i - \bm{p}_p \| \leq \delta_r , 
\end{equation}
where $\bm{p}_p$ is the position of an object of interest, the vehicle will notify neighboring vehicles by creating a detectable hidden signature. To achieve this, the vehicle switches to a \textit{hidden} virtual spring-damper model described by,
\begin{equation} \label{eq:signature_spring_force}
\bm{u}_i^{\mathcal{H}} = \ddot{\bm{p}}_i= \Big[ \kappa_{h}( l_{ip} - l^0_h ) \vec{\bm{d}}_{ip} - \gamma_h \dot{\bm{p}}_i \Big] \in \R^{m},
\end{equation}
where the virtual spring-damper parameters $\kappa_{h} \ne \kappa_{v}$ and/or $\gamma_h \ne \gamma_v$ are distinct from the primary network model in \eqref{eq:spring_force} to enable an identifiable dynamical signature. A vehicle $i$ that follows the hidden model \eqref{eq:signature_spring_force} removes all virtual spring interactions to neighboring vehicles and the goal from the primary network model that affect its control input. To maintain secrecy from attackers (with regards to the observance of the object of interest), a vehicle $i$ will continue to broadcast state, observed obstacle positions, and its neighbor set information to nearby vehicles as it would in nominal conditions. In this way, a malicious agent who is listening will continue to see the same type of information as before. Any manipulation of such information that does not conform with the new hidden model \eqref{eq:signature_spring_force} or with the primary model in \eqref{eq:spring_force} will be considered a cyber-attack.

The challenge that arises is that the object position $\bm{p}_p$ remains unknown to the other vehicles in the network. In comparison to the primary model \eqref{eq:spring_force}, nearby vehicles do not receive all necessary information when a vehicle $i$ follows the hidden model \eqref{eq:signature_spring_force} to monitor for consistency. This is due to constraints set in Problem \ref{problem2}, that information regarding a discovered object of interest can not be explicitly shared with the network to protect from interception by attackers.

Given that the hidden model \eqref{eq:signature_spring_force}, vehicle dynamics \eqref{eq:vehicle_dynamical_model}, and maximum sensing range $\delta_r$ are known by all vehicles, an expected vehicle velocity behavior can be leveraged as a vehicle converges toward an object (i.e., a decaying velocity magnitude). More specifically, any vehicle $i$ can recognize the hidden signature by monitoring the received velocity estimate $\hat{\bm{v}}_j^{(k)}$ behavior from a vehicle $j$ and compare it to the expected velocity decay behavior from the hidden model. Shown in Fig.~\ref{fig:VelocityDecay}(a) is an example of the differing expected velocity behavior between springs of the primary and hidden models with the corresponding distances to the object.

\begin{figure}[htb!]
\vspace{-4pt}
\begin{tabular}{cc}
\hspace{-10pt} \subfigure[\label{fig:vel} ]{\includegraphics[width = 0.24\textwidth]{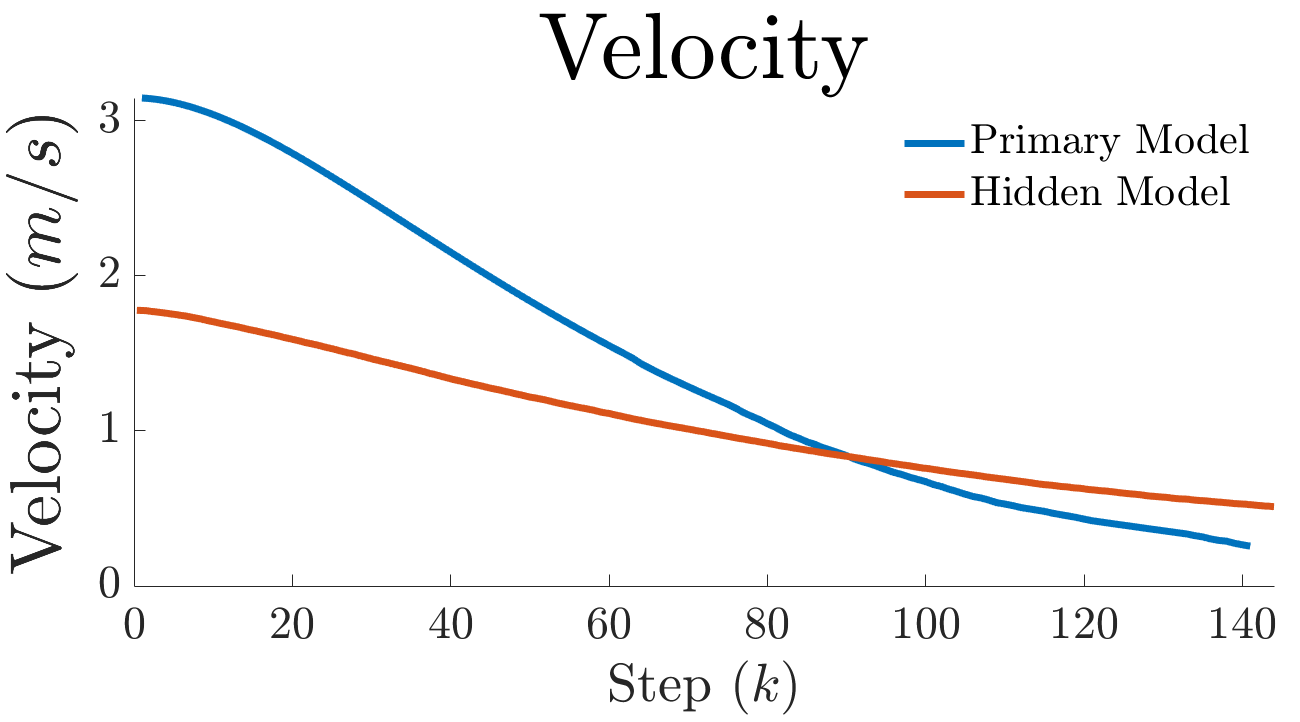}} &
\hspace{-14pt} \subfigure[\label{fig:dist} ]{\includegraphics[width = 0.24\textwidth]{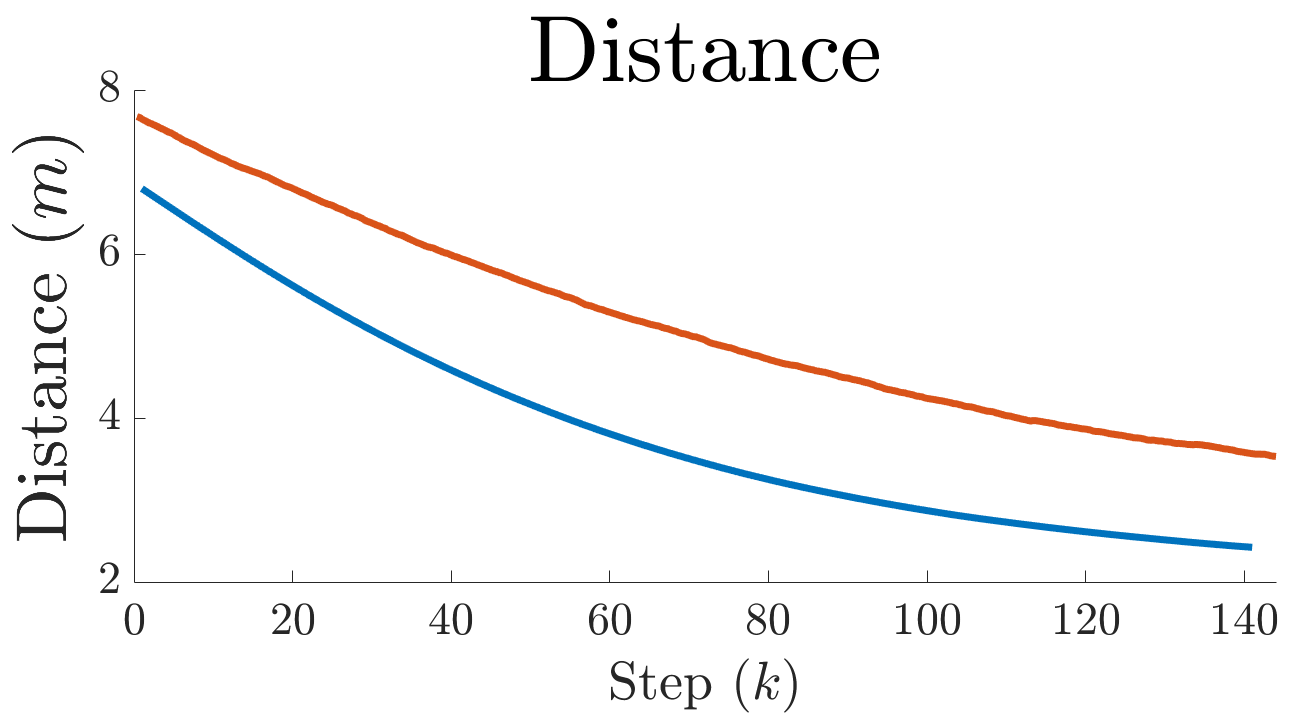}}
\end{tabular}
\vspace{-13pt}
\caption{Differing expected behavior of the (a) velocity decay, and (b) distance to the object for the two different virtual spring-damper models.}
\label{fig:VelocityDecay}
\vspace{-3pt}
\end{figure}

Given that each vehicle $i$ is making state predictions of any neighboring vehicle $j$ according to the primary network model \eqref{eq:spring_force}, an alternative action by this vehicle (i.e., utilizing the hidden model) would result in an unexpected behavior. 
Alarm rates from CUSIGN \eqref{pro:CUSIGN} on-board a vehicle $i$ that is monitoring vehicle $j$, in turn, go beyond detection bounds due to the unexpected behavior and vehicle $j$ is placed in the compromised vehicle set $\mathcal{R}_i \subset \mathcal{V}$. Next, vehicle $i$ would begin to monitor the received velocity information of vehicle $j$ to determine if its behavior follows the hidden model \eqref{eq:signature_spring_force}. A velocity prediction of vehicle $j$ by a vehicle $i$ given $j \in \mathcal{R}_i$ is made from the received velocity estimate $\| \hat{\bm{v}}_{j}^{(k)} \| $ by,
\begin{equation} \label{eq:velocity_prediction}
    \bar{v}_{ij}^{(k+1)} = h \big( \| \hat{\bm{v}}_{j}^{(k)} \|  \big),
\end{equation}
where the function $h(\cdot)$ represents the expected velocity behavior according to the hidden spring model \eqref{eq:signature_spring_force}, as shown in Fig.~\ref{fig:vel}. At each time iteration $k$, the \textit{hidden velocity residual} $\breve{r}_{ij}^{(k)}$ --- the difference between received velocity magnitudes and velocity predictions using the hidden model --- of vehicle $j$ is computed by the following,
\begin{equation} \label{eq:velocity_residual}
    \breve{r}_{ij}^{(k)} = \| \hat{\bm{v}}_{j}^{(k)} \| - \bar{v}_{ij}^{(k)} \in \R,
\end{equation}
to monitor whether vehicle $j$ is following the hidden model in \eqref{eq:signature_spring_force}. We leverage the known zero-mean Normally distributed velocity estimate provided by vehicle $j$ (see estimation error covariance in \cite{CUSUM1,Paul_ACC,Paul_IFAC}) when characterizing the hidden velocity residual. An assumption can be made such that the received velocity estimate information from vehicle $j$ is also approximately zero-mean Normally distributed around the expected velocity decay behavior in $h \big( \| \hat{\bm{v}}_{j}^{(k)} \| \big)$, \textit{only} if vehicle $j$ is following the hidden model. In this scenario, the hidden velocity residual \eqref{eq:velocity_residual} is expressed as a random variable that presents the following characteristics:
\begin{equation}
\label{eq:difference_probs}
    \begin{split}
    \Pr\big( \breve{r}_{ij}^{(k)} < 0 \big) &= \breve{p}_- = 0.5, \\[-2pt]
    \Pr\big( \breve{r}_{ij}^{(k)} > 0 \big) &= \breve{p}_+ = 0.5,
    \end{split}
    \vspace{-1pt}
\end{equation}
where the probability of the hidden velocity residual being greater or less than zero is equal. A random variable with characteristics that follow \eqref{eq:difference_probs} should present an expected sign switching rate behavior (i.e. how frequently $\breve{r}_{ij}^{(k)}$ changes signs) in accordance to the probabilities in \eqref{eq:difference_probs}. To capture the rate of sign switching, we leverage an alarm that is triggered (i.e., $\psi_{ij}^{(k)} = 1$) when a sign switch occurs at a time $k$. The procedure to trigger a sign switching alarm follows: 
\begin{equation} \label{eq:Sign_Switch}
\begin{split}
    \psi_{ij}^{(k)} = \left\{ \begin{array}{rl}
	1, & \text{if } \mathrm{sgn}\big( \breve{r}_{ij}^{(k)} \big) = -\mathrm{sgn}\big( \breve{r}_{ij}^{(k-1)} \big), \\[1pt]
    0, & \text{otherwise}. 
    \end{array} \right.
\end{split}
\end{equation} 

The sign switching alarm $\psi_{ij}^{(k)} \in \{ 0, 1 \}$ is then sent into the MRE algorithm \eqref{eq:MRE_algorithm} to provide an updated run-time estimate of the hidden signature sign switching alarm rate $\hat{H}_{ij}^{(k)} \in [0,1]$ at time instance $k$.

\begin{figure*}[htb!]
\begin{tabular}{cccc}
\hspace{-4pt} \subfigure[\label{fig:first} ]{\setlength{\fboxsep}{0pt}\fbox{\includegraphics[width = 0.235\textwidth]{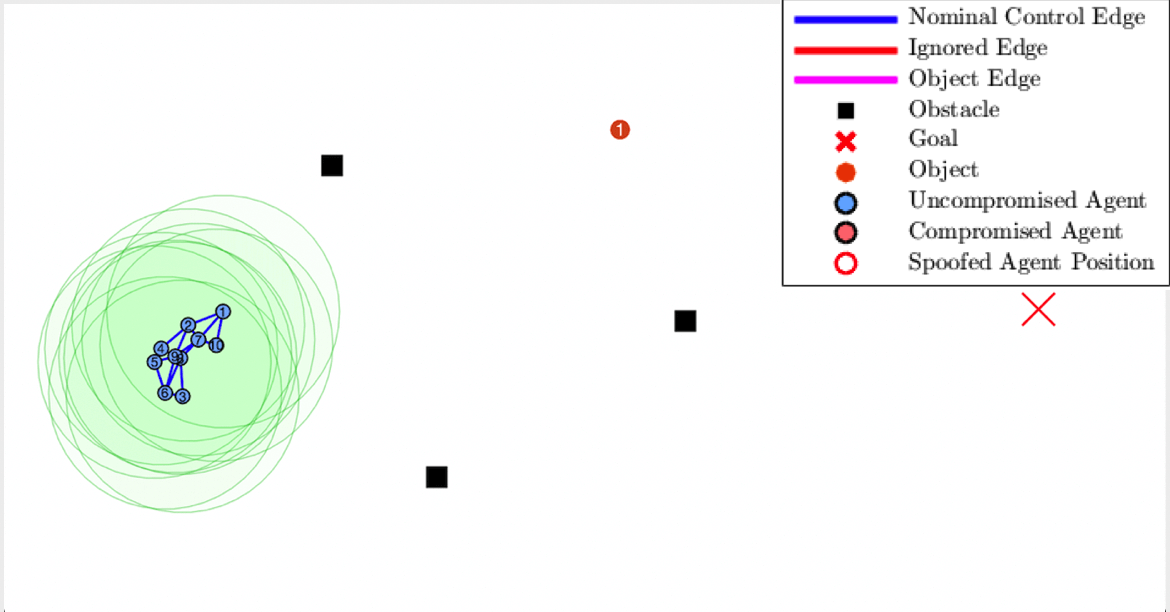}}} &
\hspace{-10pt} \subfigure[\label{fig:second} ]{\setlength{\fboxsep}{0pt}\fbox{\includegraphics[width = 0.235\textwidth]{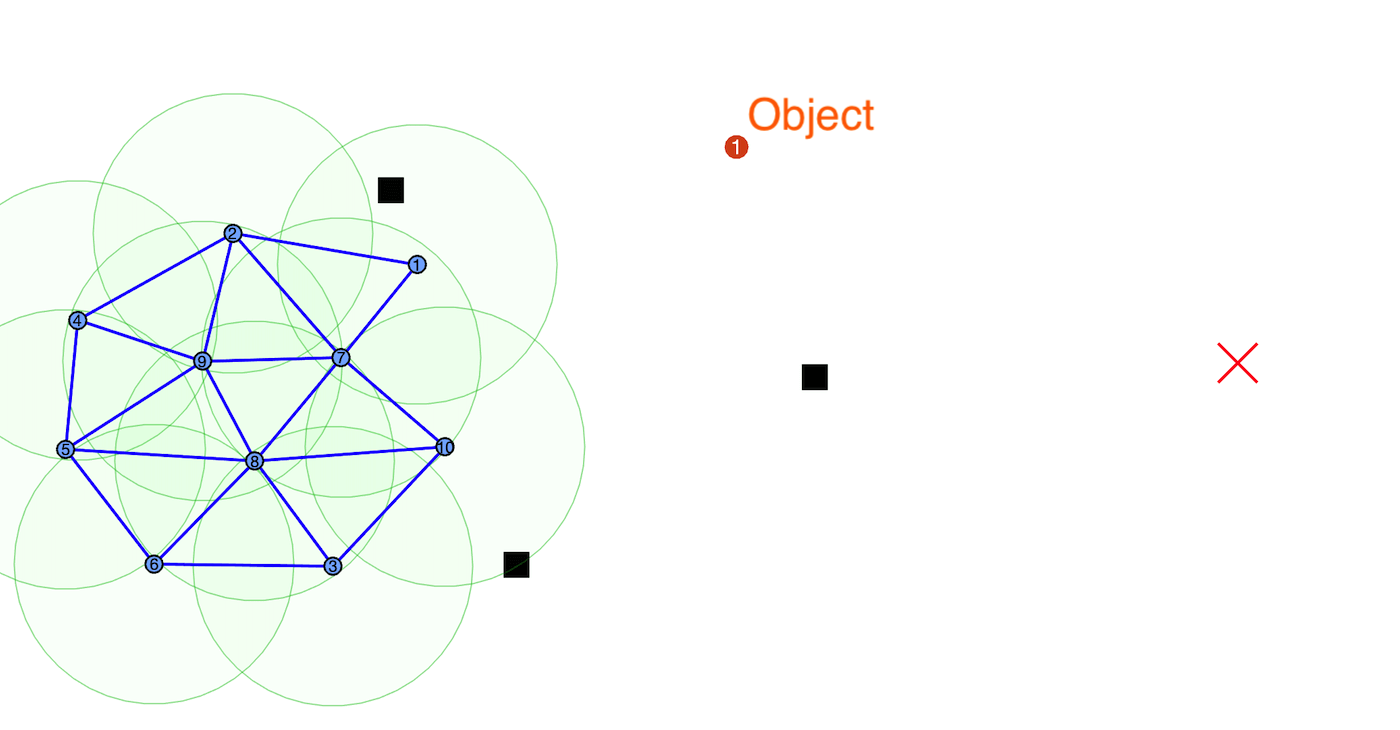}}} &
\hspace{-10pt} \subfigure[\label{fig:third} ]{\setlength{\fboxsep}{0pt}\fbox{\includegraphics[width = 0.235\textwidth]{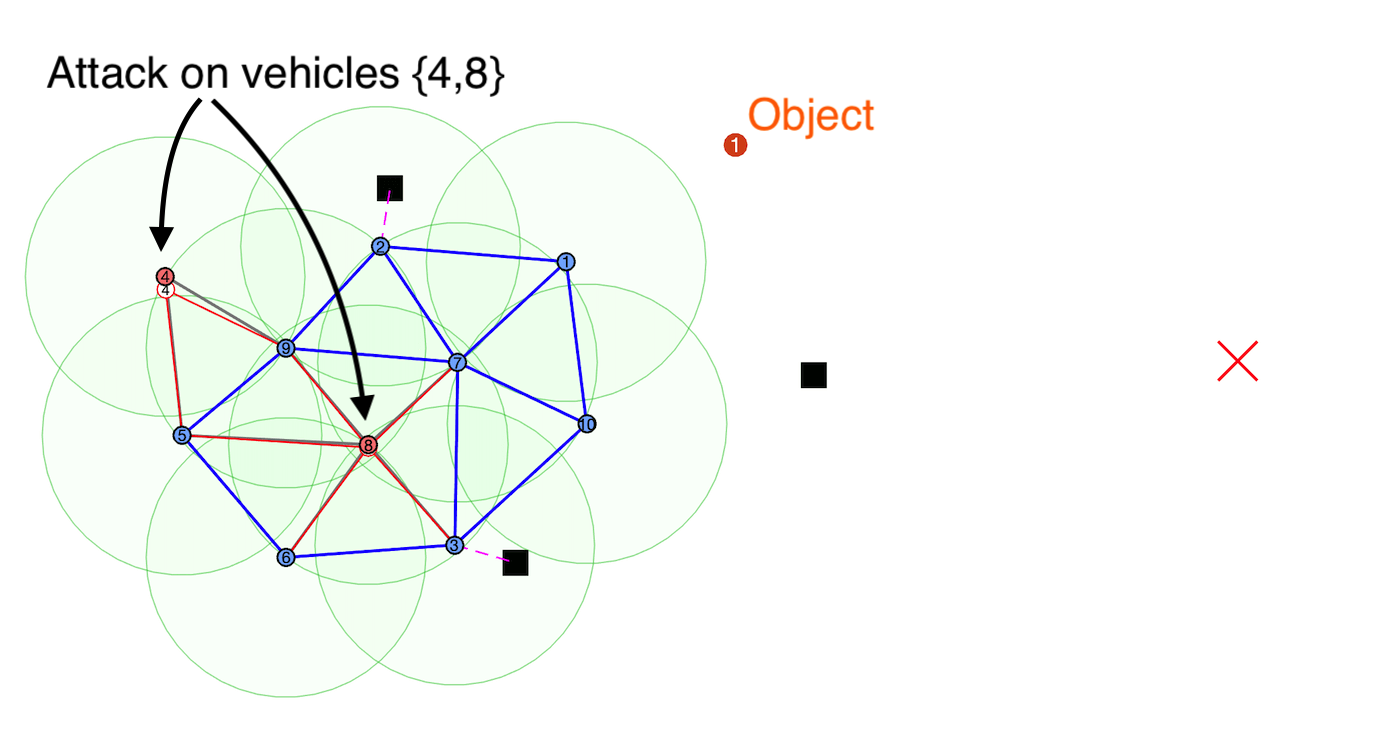}}} &
\hspace{-10pt} \subfigure[\label{fig:fourth} ]{\setlength{\fboxsep}{0pt}\fbox{\includegraphics[width = 0.235\textwidth]{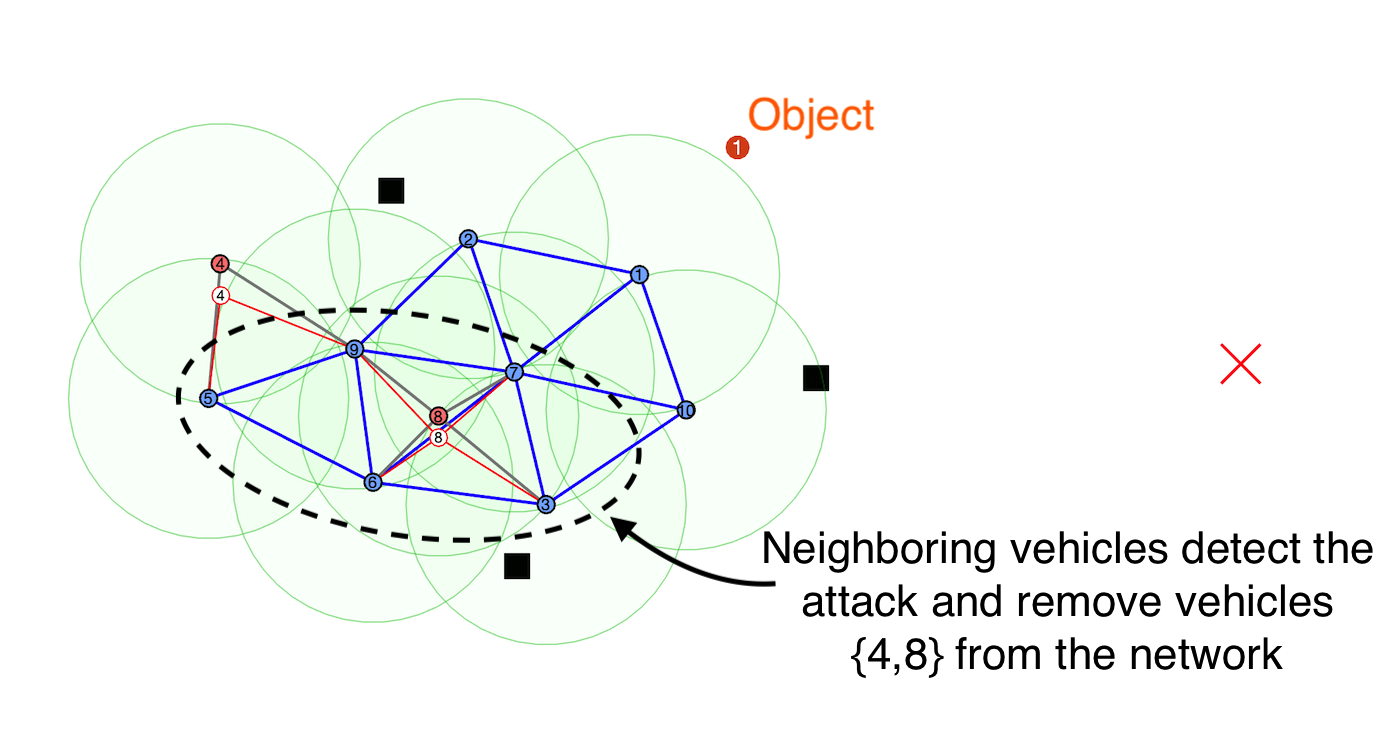}}} \\[-2pt]
\hspace{-4pt} \subfigure[\label{fig:five} ]{\setlength{\fboxsep}{0pt}\fbox{\includegraphics[width = 0.235\textwidth]{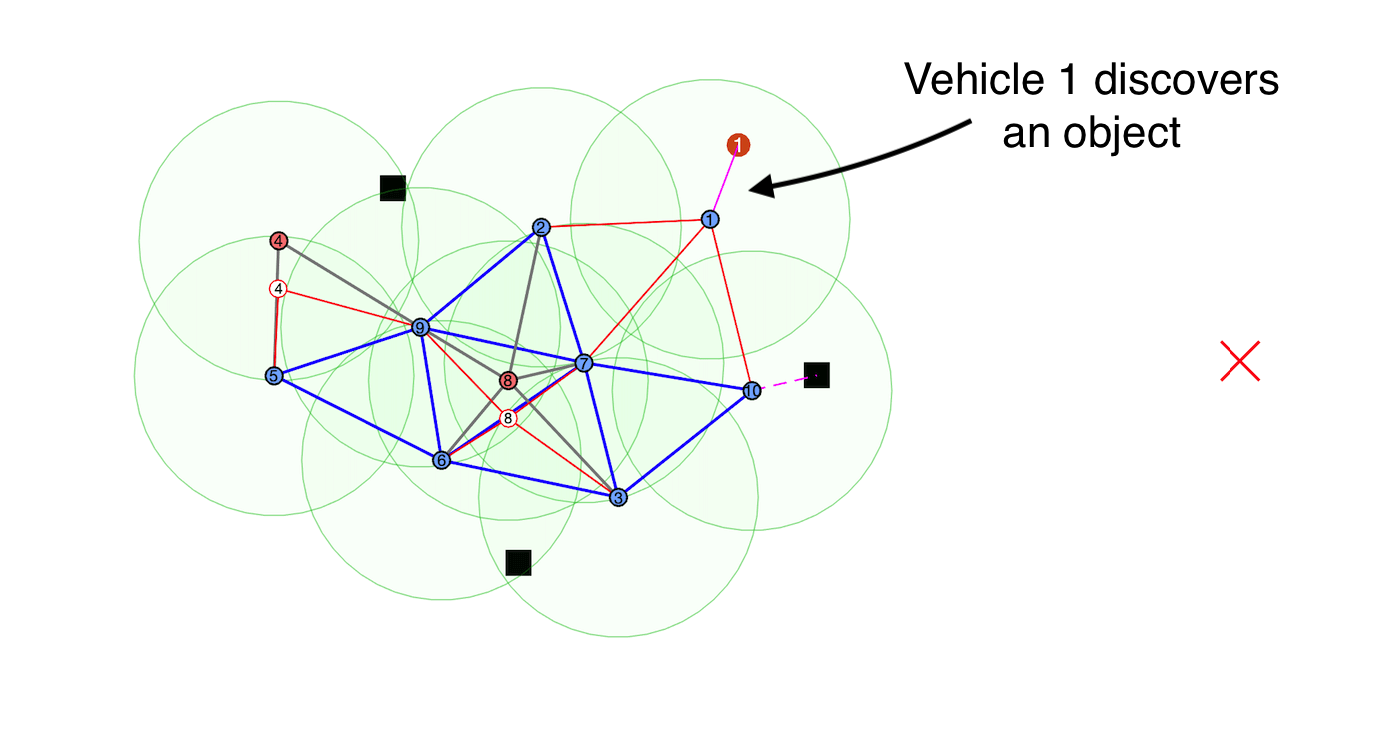}}} &
\hspace{-10pt} \subfigure[\label{fig:six} ]{\setlength{\fboxsep}{0pt}\fbox{\includegraphics[width = 0.235\textwidth]{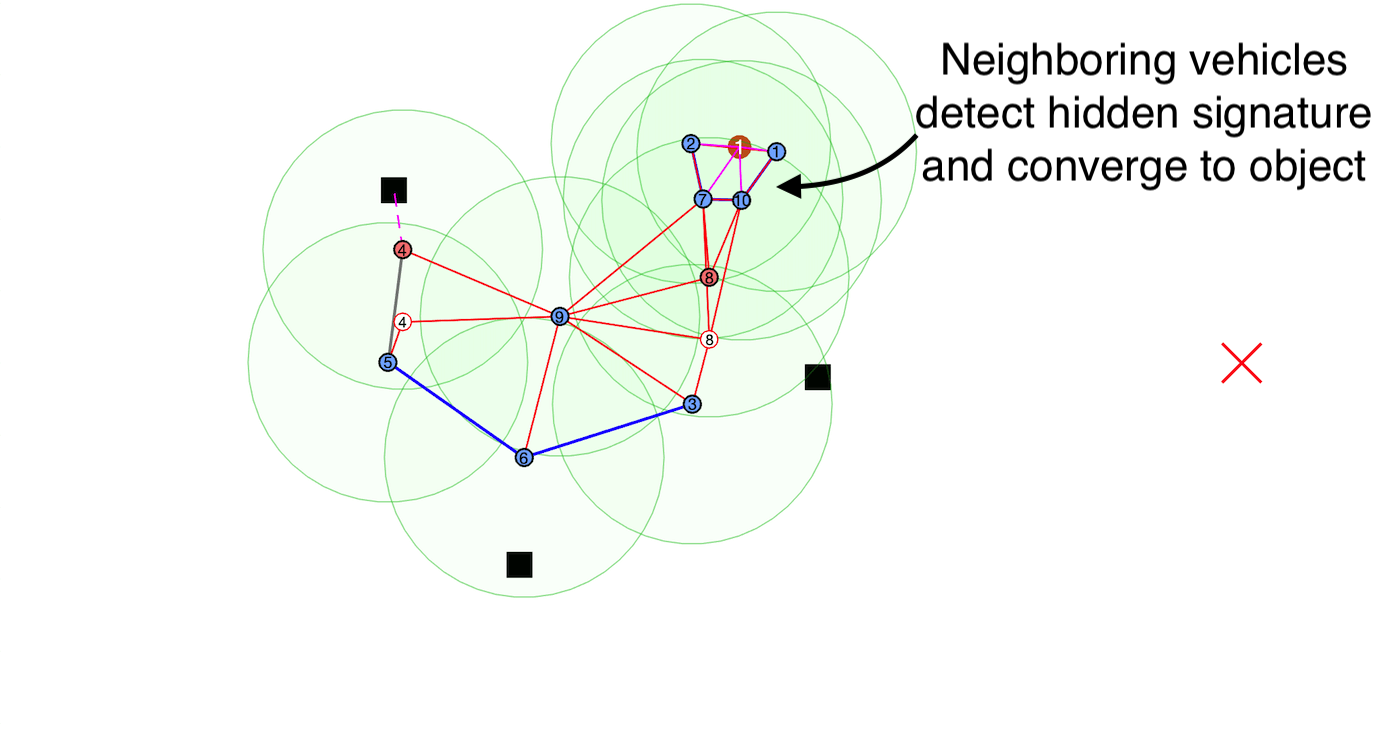}}} &
\hspace{-10pt} \subfigure[\label{fig:seven} ]{\setlength{\fboxsep}{0pt}\fbox{\includegraphics[width = 0.235\textwidth]{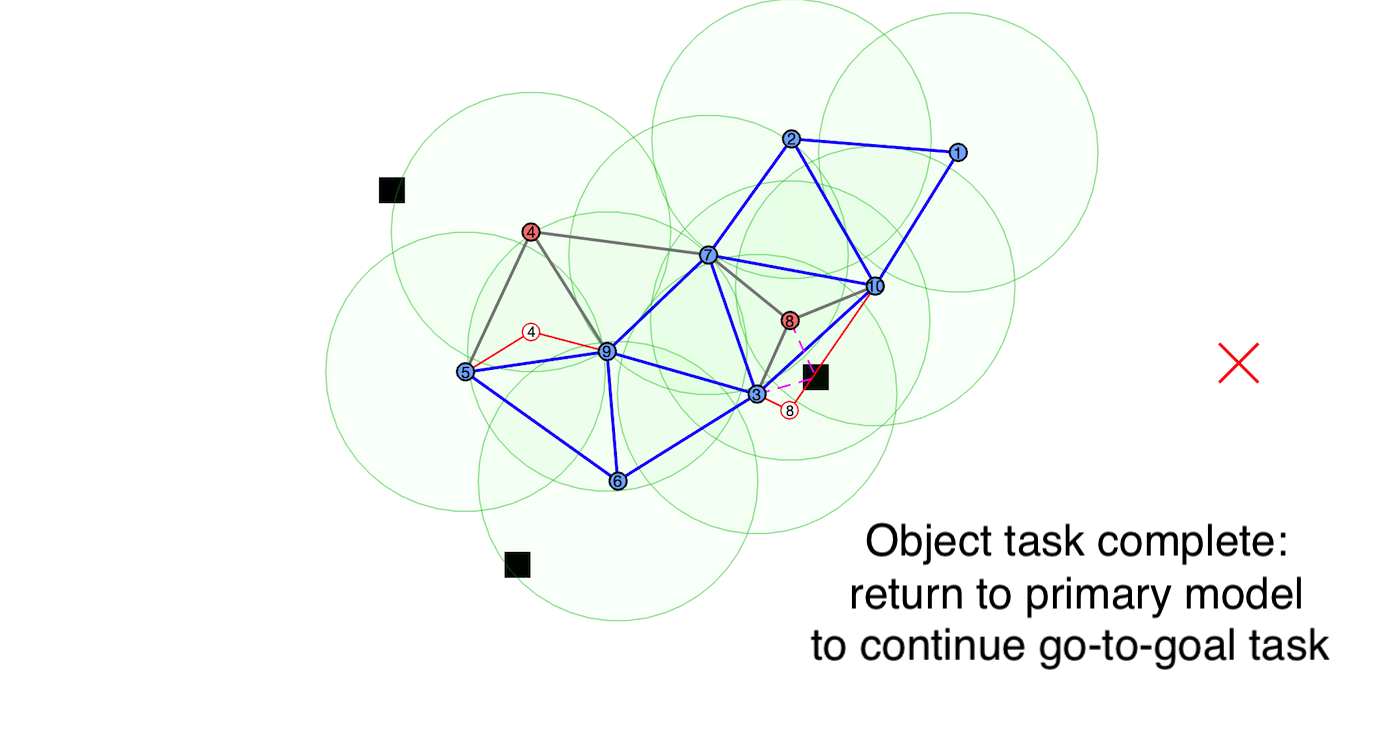}}} &
\hspace{-10pt} \subfigure[\label{fig:eight} ]{\setlength{\fboxsep}{0pt}\fbox{\includegraphics[width = 0.235\textwidth]{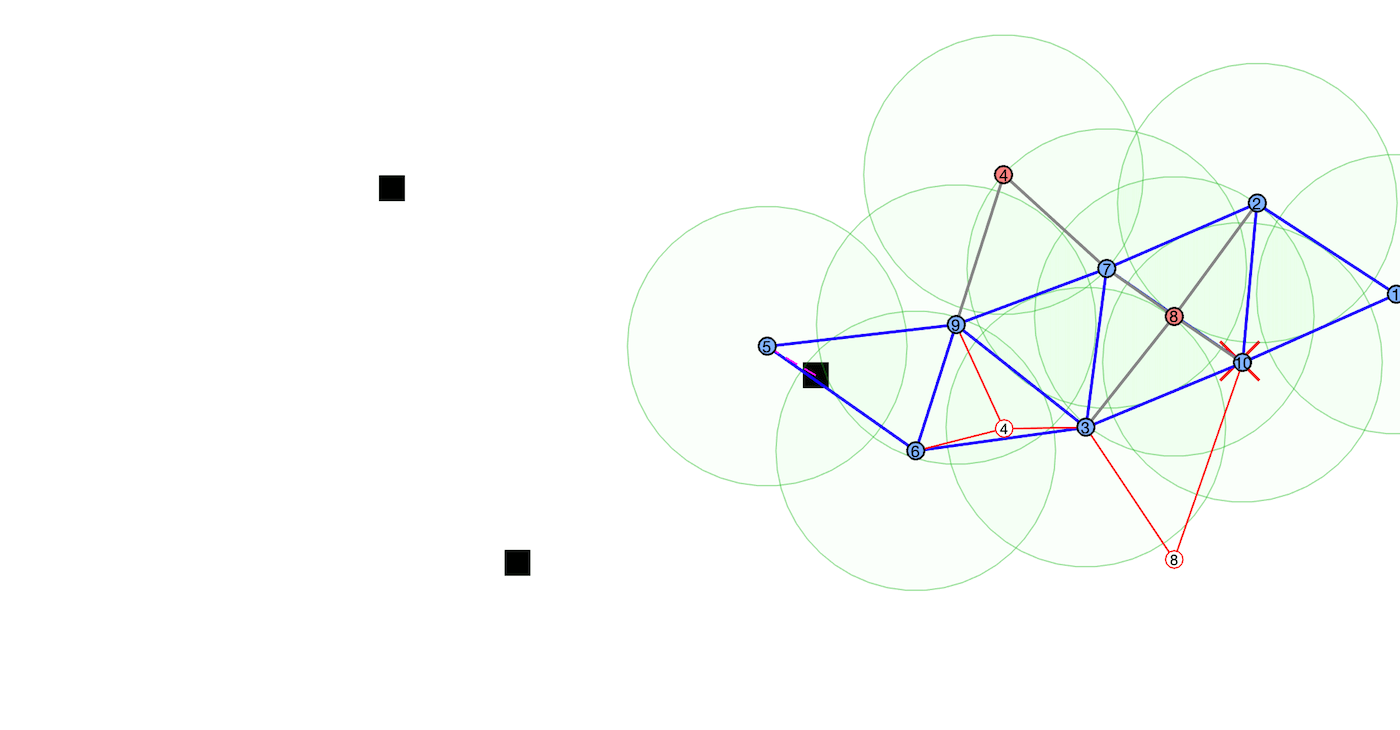}}}
\end{tabular}
\vspace{-10pt}
\caption{A network of $N = 10$ vehicles resiliently navigate through an obstacle-filled (black squares) environment to a desired goal (red `X'). Vehicles converge to an object (orange disk) as it comes within their viewing range (green disks) or if they detect the hidden signature from neighboring vehicles.}
\label{fig:Simulation}
\vspace{-12pt}
\end{figure*}

\begin{lemma} \label{lem:RoR_expected_AR}
   Given a vehicle $j$ that is following the hidden model \eqref{eq:signature_spring_force} while being monitored by vehicle $i$, the expected sign switching rate to signify random behavior is $\E[{H}] = \frac{1}{2}$.
\end{lemma}
\vspace{1pt}

\begin{proof} \label{proof:RoR_expected_AR}
    We first examine the asymptotic distribution of the expected number of observed runs $\E[U]$ from the Wald-Wolfowitz runs test \cite{wald1940}. Then, we convert $\E[U]$ over a defined sequence length to a rate described by how frequently runs should occur (i.e., how often sign switching occurs) by leveraging the known characteristics of the probabilities $\breve{p}_+$ and $\breve{p}_-$, such that the random variable follows $\E[H] = 2 \breve{p}_+\breve{p}_- = \frac{1}{2},$ thus concluding the proof.
\end{proof}

\begin{lemma} \label{lem:RoR_variance_AR}
    The expected variance of the sign switching rate $\hat{H}_{ij}^{(k)}$ for a vehicle $j$ that follows the hidden model \eqref{eq:signature_spring_force} while monitored by vehicle $i \in \mathcal{V}$ is $\mathrm{Var}[{H}] = \frac{1}{4(2 \ell -1)}$.
\end{lemma}

\begin{proof} \label{proof:RoR_variance_AR}
    Let the expectation of a sign switch be modeled by a Binomial distribution where the probability of success (i.e., sign switch) is $\E[H]$. By normal approximation and utilizing MRE \eqref{eq:MRE_algorithm} for sign switching rate estimation, the random variable follows a normal distribution with variance $\mathrm{Var}[H] = \frac{\E[H](1-\E[H])}{2\ell-1} = \frac{1}{4(2 \ell -1)}$, concluding the proof.
\end{proof}

The following corollary provides bounds of $\hat{H}_{ij}^{(k)}$ to satisfy an expected behavior to detect the hidden model signature.
\begin{corollary} \label{cor:RoR_bounds}
Given the sequence of hidden velocity residuals $\breve{r}_{ij}^{(k)}$, hidden signature detection occurs by the sign switching alarm rate when $ \Psi_- \leq \hat{H}_{ij}^{(k)} \leq \Psi_+$ is satisfied.
\end{corollary}

\begin{proof}
A proof can be obtained by leveraging confidence intervals within a Normal Distribution. Due to page limitations, we omit the proof.
\end{proof}

To summarize Corollary \ref{cor:RoR_bounds}, when the sign switching alarm rate for the detection of the hidden signature satisfies,
\begin{equation} \label{eq:Signature_detect}
    \hat{H}_{ij}^{(k)} \in [\Psi_-, \Psi_+]  \longrightarrow  \textit{Signature Detection},
\end{equation} 
vehicle $i$ detects a hidden signature behavior in $j$. Vehicle $i$ reacts by estimating the position of the object by leveraging the training set (i.e., expected hidden model behavior) mapping $f : \R^2 \rightarrow \R$ in Fig.~\ref{fig:VelocityDecay} that maps the received velocity estimate of vehicle $j$ to its distance to the object by,
\begin{equation} \label{eq:vel_to_dist_mapping}
    \hat{d}^{(k)}_{p,ij} = f( \| \hat{\bm{v}}_j^{(k)} \| ).
\end{equation}

The position of the object $\bm{p}_p$ is then estimated by vehicle $i$ from the received position information of vehicle $j$ by,
\begin{equation} \label{eq:object_estimate}
    \hat{\bm{p}}_{p,ij}^{(k)} = \hat{\bm{p}}_{j}^{(k)} + \hat{d}^{(k)}_{p,ij} \vec{\bm{d}}(\| \hat{\bm{v}}_j^{(k)} \|), 
\end{equation}
where $ \vec{\bm{d}}(\cdot)$ is a unit vector indicating velocity direction of vehicle $j$. Once the object position coordinates $\hat{\bm{p}}_{p,ij}^{(k)}$ have been estimated, vehicle $i$ then detaches virtual springs from its neighbors and goal to converge to the object of interest by also following the hidden model \eqref{eq:Signature_detect}. Vehicles within the network will continue to converge toward the point of interest until its task is completed. Upon completion, all vehicles involved with the hidden task return to their normal control behavior with the primary network model in \eqref{eq:spring_force}.

\section{Results} \label{sec:results}

Our approach is validated with Matlab simulations and experiments using swarms of Turtlebot 2 robots, as shown in the provided video. In both case studies, the vehicle networks leverage a primary network model \eqref{eq:spring_force} to perform a go-to-goal task and a hidden model \eqref{eq:signature_spring_force} to covertly notify surrounding vehicles when an object of interest has been discovered. Furthermore, vehicles are subject to MITM cyber-attacks that are attempting to hijack the network to an undesired state.

\subsection{Simulation} \label{sec:simulations}

For the simulation case study, we consider $N = 10$ vehicles treated as double integrator point-masses navigating in an $x$-$y$ plane. A sequence of snapshots are presented in Fig.~\ref{fig:Simulation} showing the network of vehicles resiliently navigating through an obstacle filled environment. From Fig.~\ref{fig:third}-(h), vehicles $\{4,8 \}$ (red circles) are subject to MITM attacks that falsify position information with the intention of hijacking the network. In Fig.~\ref{fig:five} vehicle $1$ discovers an object and then switches to the hidden model \eqref{eq:signature_spring_force} to covertly notify neighboring vehicles about the discovery. In turn, neighboring vehicles detect this hidden signature \eqref{eq:Signature_detect}, estimate the position of the object \eqref{eq:object_estimate}, then also switch to the hidden model to converge to the object. CUSIGN and sign switching alarm rates from the perspective of vehicle $i=7$ are provided in Fig.~\ref{fig:Simulation_AR} while monitoring neighboring vehicles for primary and hidden model behaviors. Alarm rates for CUSIGN \eqref{pro:CUSIGN} monitoring vehicles $\{ 1,2,4,8,10 \}$ travel beyond detection bounds indicated by red dashed lines. However, shown in Fig.~\ref{fig:Simulation_AR}(b), the hidden signature alarm rates \eqref{eq:Signature_detect} for vehicles $\{ 1,2,10 \}$ satisfy the hidden model detection bounds to signify consistent behavior is occurring with respect to the hidden model \eqref{eq:signature_spring_force}, thus deeming these vehicles trustworthy. Alternatively, vehicles $\{4,8 \} \in \mathcal{R}_7$ are treated compromised as a hidden signature was not detected from their motion.

\begin{figure}[h!]
\begin{tabular}{cc}
\hspace{-9pt} \subfigure[\label{fig:AR_first} ]{\includegraphics[width = 0.241\textwidth]{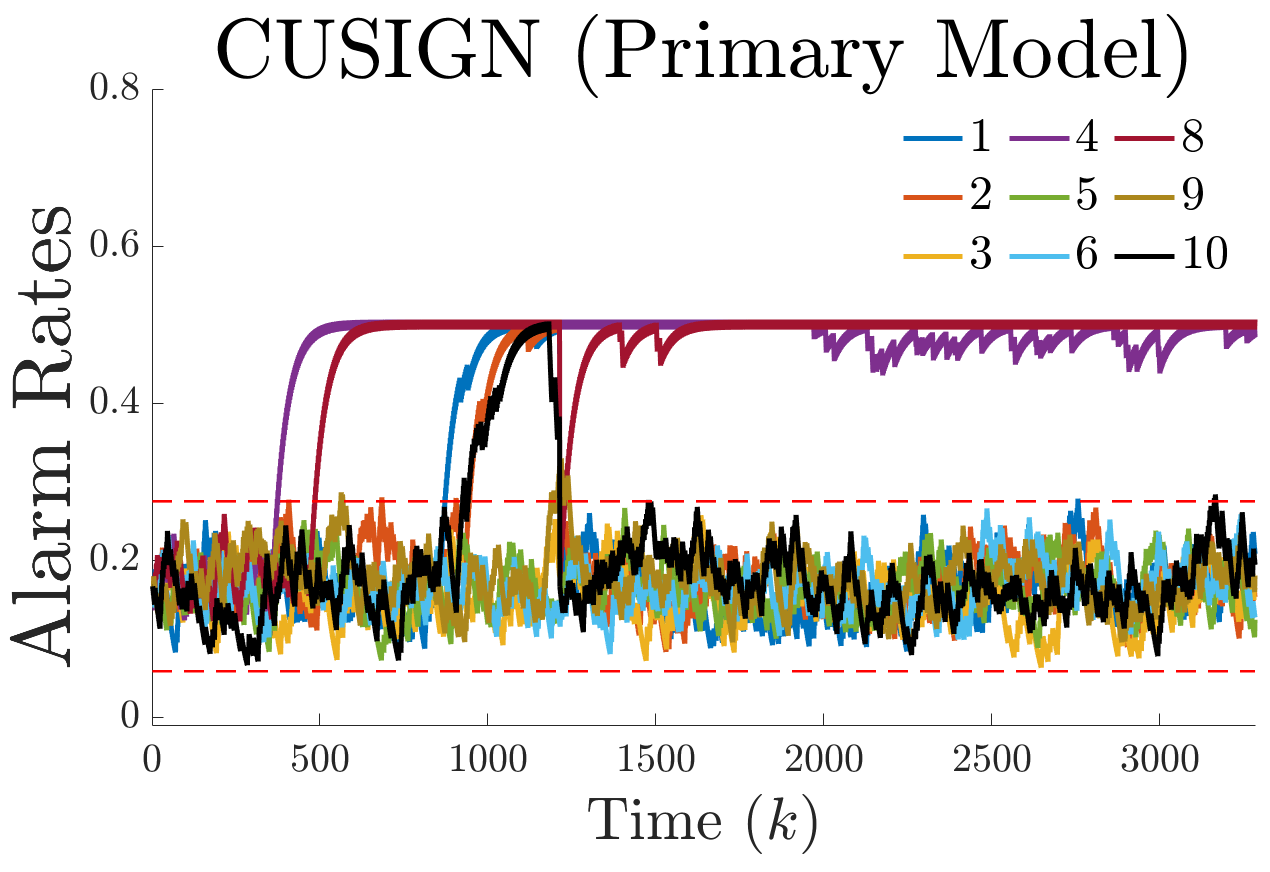}} &
\hspace{-16pt} \subfigure[\label{fig:AR_second} ]{\includegraphics[width = 0.241\textwidth]{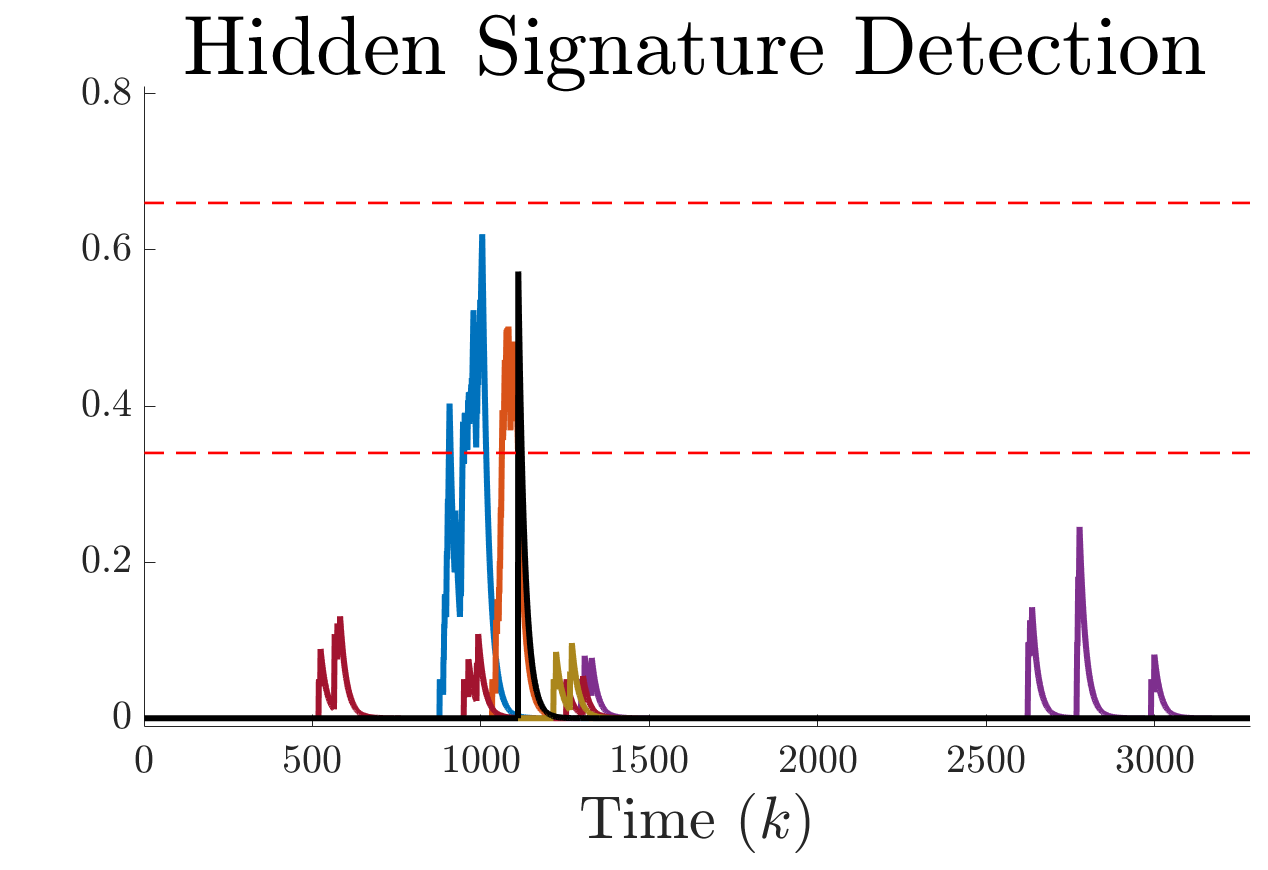}}
\end{tabular}
\vspace{-12pt}
\caption{(a) CUSIGN detection alarm rates from the perspective of vehicle $i = 7$ and (b) sign switching alarm rates for hidden signature detection.}
\label{fig:Simulation_AR}
\vspace{-9pt}
\end{figure}

\subsection{Experiment} \label{sec:experiments}

\begin{figure*}[t!hb]
\begin{tabular}{ccccccc}
\hspace{-11pt} \renewcommand{\thesubfigure}{}
\subfigure[\label{fig:gopro1} ]{\includegraphics[width = 0.13\textwidth]{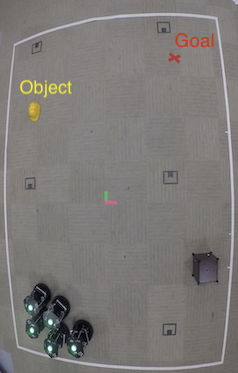}} &	
\hspace{-14pt} \renewcommand{\thesubfigure}{}
\subfigure[\label{fig:gopro2} ]{\includegraphics[width = 0.13\textwidth]{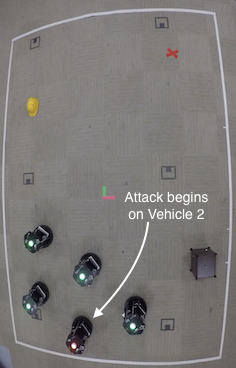}} &
\hspace{-14pt} \renewcommand{\thesubfigure}{}
\subfigure[\label{fig:gopro3} ]{\includegraphics[width = 0.13\textwidth]{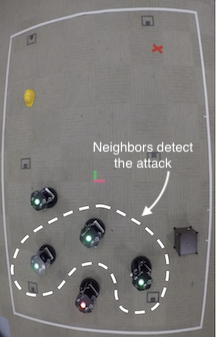}} &
\hspace{-14pt} \renewcommand{\thesubfigure}{}
\subfigure[\label{fig:gopro4} ]{\includegraphics[width = 0.13\textwidth]{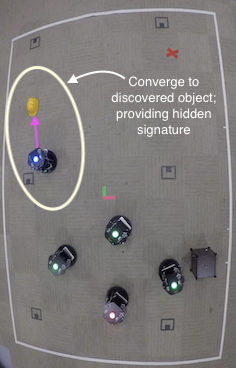}} &
\hspace{-14pt} \renewcommand{\thesubfigure}{}
\subfigure[\label{fig:gopro5} ]{\includegraphics[width = 0.13\textwidth]{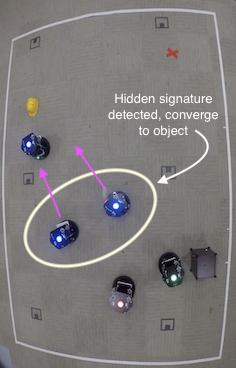}} &
\hspace{-14pt} \renewcommand{\thesubfigure}{}
\subfigure[\label{fig:gopro6} ]{\includegraphics[width = 0.13\textwidth]{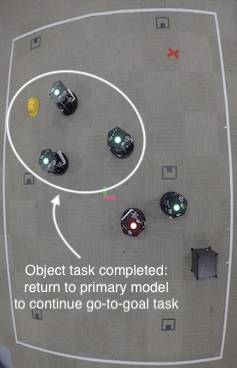}} &
\hspace{-14pt} \renewcommand{\thesubfigure}{}
\subfigure[\label{fig:gopro7} ]{\includegraphics[width = 0.13\textwidth]{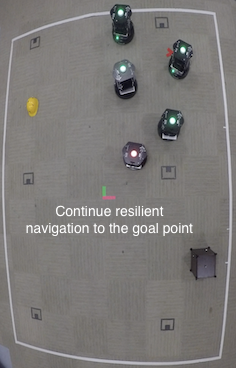}} \\[-4pt]
\hspace{-11pt} \renewcommand{\thesubfigure}{(a)}
\subfigure[\label{fig:matlab1} ]{\setlength{\fboxsep}{0pt}\fbox{\includegraphics[width = 0.1291\textwidth]{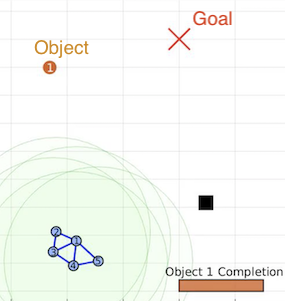}}} &
\hspace{-14pt} \renewcommand{\thesubfigure}{(b)}
\subfigure[\label{fig:matlab2} ]{\setlength{\fboxsep}{0pt}\fbox{\includegraphics[width = 0.13\textwidth]{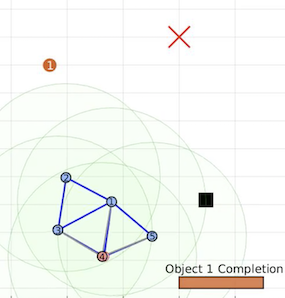}}} &
\hspace{-14pt} \renewcommand{\thesubfigure}{(c)}
\subfigure[\label{fig:matlab3} ]{\setlength{\fboxsep}{0pt}\fbox{\includegraphics[width = 0.13\textwidth]{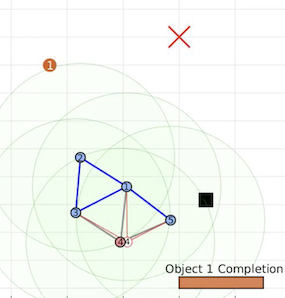}}} &
\hspace{-14pt} \renewcommand{\thesubfigure}{(d)}
\subfigure[\label{fig:matlab4} ]{\setlength{\fboxsep}{0pt}\fbox{\includegraphics[width = 0.13\textwidth]{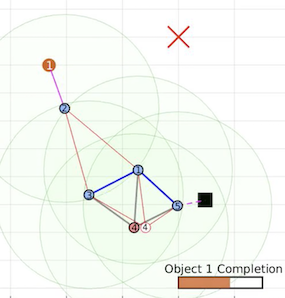}}} &
\hspace{-14pt} \renewcommand{\thesubfigure}{(e)}
\subfigure[\label{fig:matlab5} ]{\setlength{\fboxsep}{0pt}\fbox{\includegraphics[width = 0.13\textwidth]{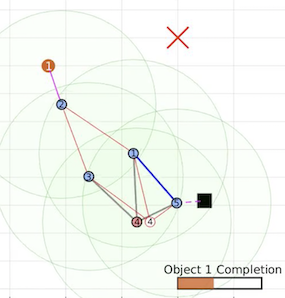}}} &
\hspace{-14pt} \renewcommand{\thesubfigure}{(f)}
\subfigure[\label{fig:matlab6} ]{\setlength{\fboxsep}{0pt}\fbox{\includegraphics[width = 0.13\textwidth]{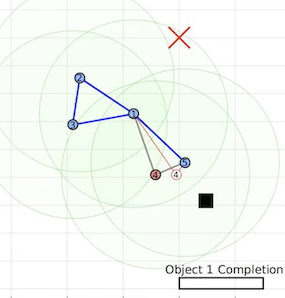}}} &
\hspace{-14pt} \renewcommand{\thesubfigure}{(g)}
\subfigure[\label{fig:matlab7} ]{\setlength{\fboxsep}{0pt}\fbox{\includegraphics[width = 0.13\textwidth]{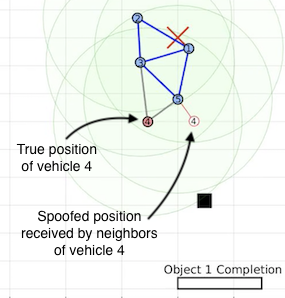}}}
\end{tabular}
\vspace{-4pt}
\caption{An experiment showing a network of $N = 5$ TurtleBot 2 robots resiliently navigating to a goal (red `X'). Vehicle $2$ discovers an object (yellow helmet) as comes within its sensing range (depicted by the green translucent circle), then provides a hidden signature behavior for nearby vehicles to recognize as it converges to the object. Neighboring vehicles also converge to the object of interest upon detection of this hidden signature.} 
\label{fig:Experiment}
\vspace{-11pt}
\end{figure*}

Experimental validations are performed on $N = 5$ TurtleBot 2 differential-drive robots performing a go-to-goal operation within a lab environment. Snapshots of this experiment are presented in Fig.~\ref{fig:Experiment} capturing the following sequence of events; the initial vehicle positions (Fig.~\ref{fig:matlab1}), vehicle $2$ discovering the object (Fig.~\ref{fig:matlab4}), neighboring vehicles converge toward the object after detecting the hidden signature from vehicle $2$ (Fig. \ref{fig:matlab5}), and the network continuing to the goal once the object ``task" has been completed (Fig. \ref{fig:matlab6}-(g)). During the simulation, communication broadcasts from vehicle $j = 4$ are corrupted with false position data that attempt to drive the system to an undesirable location, but the CUSIGN detector finds these stealthy attacks, allowing the network to resiliently perform the operation.  In Fig.~\ref{fig:Experiment_results}, alarm rates that are monitoring the primary \eqref{eq:spring_force} and hidden \eqref{eq:signature_spring_force} models throughout the experiment show vehicle $1$ detecting the compromised vehicle $4$, as well as detecting the hidden signature from vehicle $2$.

\begin{figure}[htb!]
\begin{tabular}{cc}
\hspace{-9pt} \subfigure[\label{fig:exp_alarm_rates} ]{\includegraphics[width = 0.24\textwidth]{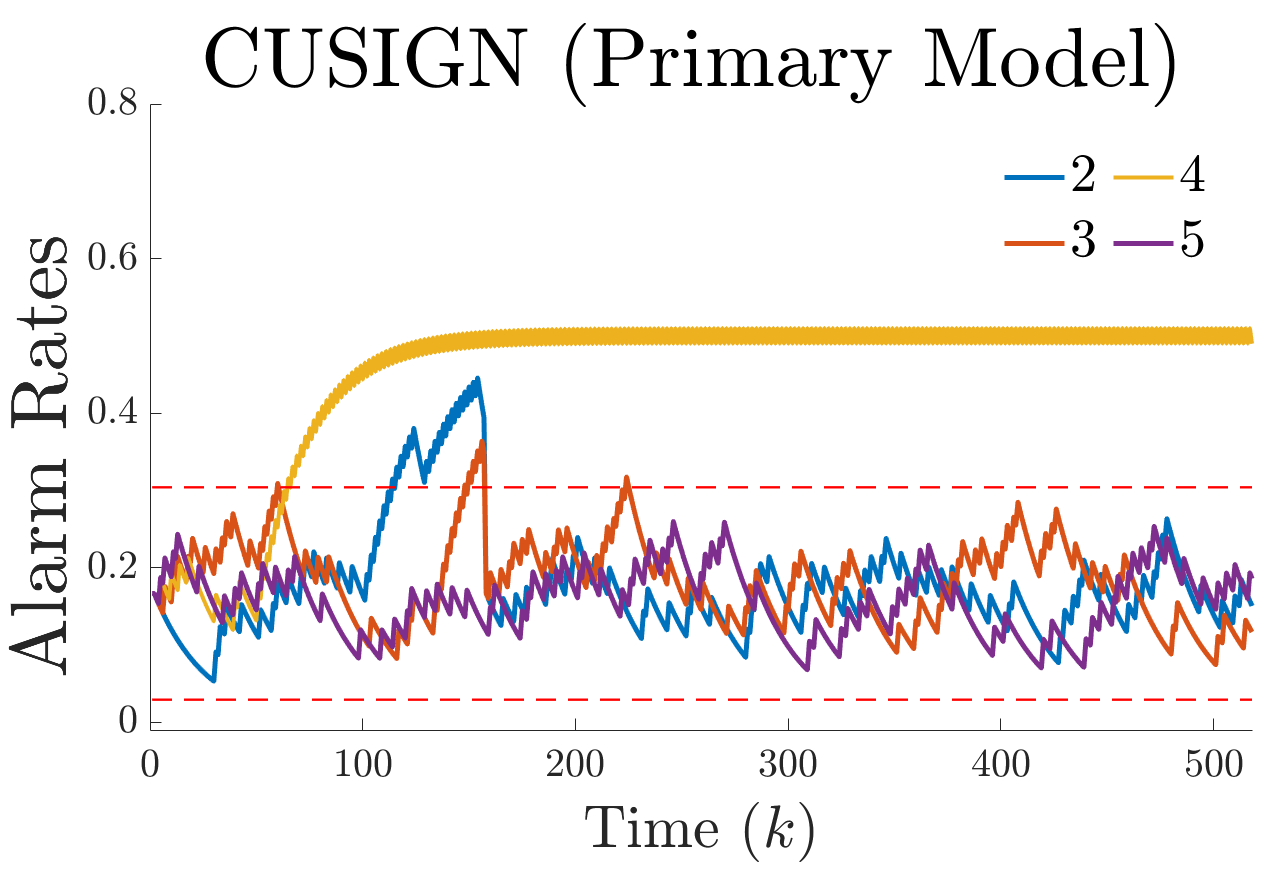}} &
\hspace{-15pt} \subfigure[\label{fig:exp_detect_signature} ]{\includegraphics[width = 0.24\textwidth]{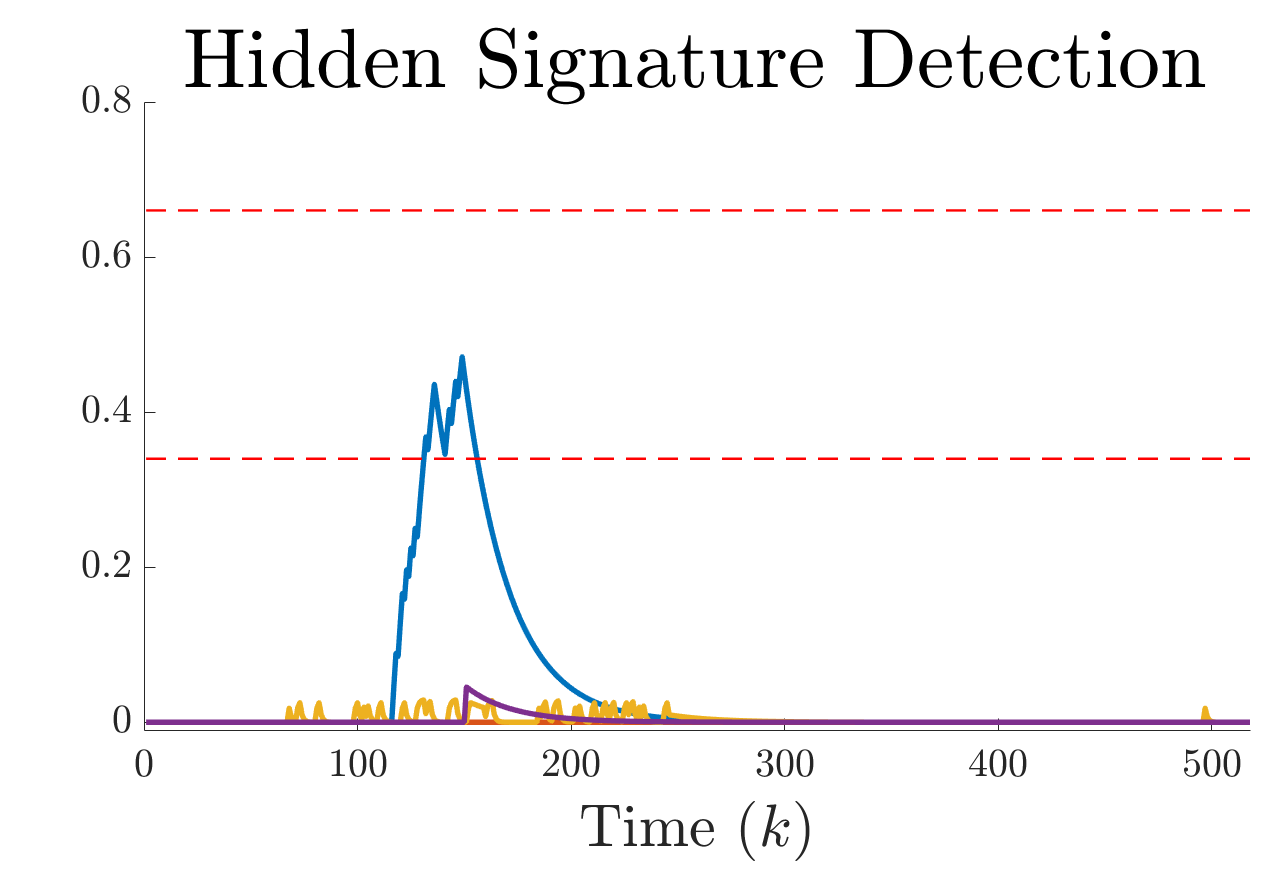}}
\end{tabular}
\vspace{-10pt}
\caption{(a) CUSIGN and (b) sign switching alarm rates from the perspective of vehicle $i = 1$. Detection bounds are indicated by red dashed lines.}
\label{fig:Experiment_results}
\vspace{-12pt}
\end{figure}
\section{Conclusions} \label{sec:conclusion}

In this paper we have proposed a decentralized framework for a network of homogeneous vehicles to resiliently perform desired operations. Vehicles are able to distinguish between received inconsistent information from neighboring vehicles due to man-in-the-middle attacks and hidden model behaviors that provide a detectable signature to implicitly pass safety-critical information. To detect stealthy attacks and the hidden signature, we leverage randomness-based detection techniques ---Cumulative Sign (CUSIGN) and sign switching rate--- to identify whether vehicles are following a primary or hidden network model. In our future work we plan to: i) extend the current approach by investigating the effects of different attack classes/models and ii) develop an adaptive approach for the virtual spring parameters to conform to changing network or environmental conditions.
\section{Acknowledgement}

This work is based on research supported by NSF under grant number \#1816591 and ONR under agreement number N000141712012. The authors would like to thank Rahul Peddi and Shijie Gao for assisting with the experiments.

\bibliographystyle{IEEEtran}
\bibliography{References}

\end{document}